\documentclass[journal]{IEEEtran}
\ifCLASSINFOpdf
\usepackage[pdftex]{graphicx}
\else
\fi
\usepackage[cmex10]{amsmath}
\usepackage{amssymb}
\usepackage{amsthm}
\usepackage{color}

\newtheorem{defi}{Definition}
\newtheorem{theorem}{Theorem}[section]

\usepackage{caption}
\usepackage{subcaption}
\usepackage{tikz}
\newcommand{\mxfifo}{\ensuremath{\mathit{fifo}}}
\newcommand{\mxatr}{\ensuremath{\mathit{atr}}}
\newcommand{\mxitr}{\ensuremath{\mathit{itr}}}
\newcommand{\mxfiforate}{\ensuremath{\mathit{fiforate}}}
\newcommand{\mxtokrate}{\ensuremath{\mathit{tokrate}}}

\newcommand{\mxfalse}{\ensuremath{\mathit{false}}}
\newcommand{\mxtrue}{\ensuremath{\mathit{true}}}
\newcommand{\mxbtoi}{\ensuremath{\mathit{BTOI}}}
\newcommand{\mxparent}{\ensuremath{\mathit{parent}}}
\newcommand{\mxcontrol}{\ensuremath{\mathit{cport}}}
\newcommand{\mxdelay}{\ensuremath{\mathit{delay}}}

\newcommand{\mxdrps}{\ensuremath{\mathit{drps}}}

\newcommand\copyrighttext{%
  \footnotesize \textcopyright 2017 IEEE. Personal use is permitted. For any other purposes, permission must be obtained from the IEEE by emailing pubs-permissions@ieee.org.}
\newcommand\copyrightnotice{%
\begin{tikzpicture}[remember picture,overlay]
\node[anchor=south,yshift=10pt] at (current page.south) {\fbox{\parbox{\dimexpr\textwidth-\fboxsep-\fboxrule\relax}{\copyrighttext}}};
\end{tikzpicture}%
}

\begin{document}

\title{PRUNE: Dynamic and Decidable Dataflow for Signal Processing on Heterogeneous Platforms}

\author{Jani~Boutellier,~\IEEEmembership{Member,~IEEE,}
        Jiahao Wu,~\IEEEmembership{Member,~IEEE,}
Heikki Huttunen,~\IEEEmembership{Member,~IEEE,}        Shuvra S. Bhattacharyya,~\IEEEmembership{Fellow,~IEEE}
\thanks{J. Boutellier is with the Laboratory
of Pervasive Computing, Tampere University of Technology, Finland,
e-mail: jani.boutellier@tut.fi.}
\thanks{J. Wu is with the Department of Electrical and Computer Engineering, University of Maryland, USA, e-mail: jiahao@terpmail.umd.edu.}
\thanks{H. Huttunen is with the Laboratory of Signal Processing, Tampere University of Technology, Finland, e-mail: heikki.huttunen@tut.fi.}
\thanks{S. S. Bhattacharyya is with the Department of Electrical and Computer Engineering, University of Maryland, USA, and the Laboratory of Pervasive Computing, Tampere University of Technology, Finland, email: ssb@umd.edu.}
\thanks{Manuscript received ----- xx, 2017; revised ----- xx, 2017.}}

%
%

\markboth{Transactions on Signal Processing,~Vol.~--, No.~--, -----~----}%
{Shell \MakeLowercase{\textit{et al.}}: Bare Demo of IEEEtran.cls for Journals}
%



\maketitle
\copyrightnotice
\begin{abstract}
The majority of contemporary mobile devices and personal computers are based on heterogeneous computing platforms that consist of a number of CPU cores and one or more Graphics Processing Units (GPUs). Despite the high volume of these devices, there are few existing programming frameworks that target full and simultaneous utilization of all CPU and GPU devices of the platform.

This article presents a dataflow-flavored Model of Computation (MoC) that has been developed for deploying signal processing applications to heterogeneous platforms. The presented MoC is dynamic and allows describing applications with data dependent run-time behavior. On top of the MoC, formal design rules are presented that enable application descriptions to be simultaneously dynamic and decidable. Decidability guarantees compile-time application analyzability for deadlock freedom and bounded memory.

The presented MoC and the design rules are realized in a novel Open Source programming environment ``PRUNE'' and demonstrated with representative application examples from the domains of image processing, computer vision and wireless communications. Experimental results show that the proposed approach outperforms the state-of-the-art in analyzability, flexibility and performance.
\end{abstract}

\begin{IEEEkeywords}
Dataflow computing, design automation, signal processing, parallel processing
\end{IEEEkeywords}

%
\IEEEpeerreviewmaketitle

\section{Introduction}

\IEEEPARstart{A}{dvances} in signal processing have enabled new technologies that greatly affect our everyday lives. Progress in wireless communications, video coding, and recently, computer vision, has provided us previously impossible applications. However, simultaneously, the signal processing behind MIMO radios, H.265 video coding and Convolutional Neural Networks has reached considerable computational complexity, even though these applications are often executed on performance-constrained mobile devices.

Enabling real-time performance for such signal processing algorithms often means resorting to computation acceleration by fixed-function ASICs (Application Specific Integrated Circuit) or by programmable accelerators such as GPUs. However, due to strict design time requirements of the industry, programmable accelerators are becoming increasingly popular compared to ASICs.

The efficiency of programmable computation accelerators is based on the fact that their architectures have been tuned to accelerate specific classes of algorithms. Unfortunately, this means that accelerators are only suitable for executing certain parts of application code, leaving the rest of the execution burden to
the general purpose CPU cores of the computation platform. Hence, both general purpose CPU cores and accelerators have a significant  role in the total application performance. To this extent, tapping the full performance potential of a heterogeneous computing platform requires a programming approach that can efficiently and simultaneously use all available CPU cores and accelerators.

The mainstream approach for programming the most popular compute accelerators of today, GPUs, involves the C-like languages CUDA from NVidia and OpenCL by Khronos. Whereas the former is intended for offloading computations to NVidia GPU devices, the latter provides a common Application Programming Interface (API) for both CPU cores and GPUs. Unfortunately, the OpenCL API operates on a low level and requires the programmer to take care of synchronization and memory transfers between devices, which is tedious and requires specialized expertise.

As it has been observed in many previous works \cite{Plishker08,Bilsen96,Thies02,Lee95}, dataflow Models of Computation provide a remarkably suitable abstraction for signal processing algorithms. Previous work \cite{Schor13} has also shown that programming frameworks based on the dataflow abstraction allow the application programmer to concentrate on developing the application, as concurrency and memory related low-level tasks are managed by the dataflow abstraction and by the programming framework.

This article describes a dataflow-flavored Model of Computation  that
\begin{itemize}
	\item captures the functionality of \textit{data dependent} signal processing algorithms,
    \item enables design time analysis for deadlock freedom and bounded memory use (decidability) through formal design rules, and
    \item provides a basis for efficient concurrent computation on heterogeneous platforms.
\end{itemize}
The MoC presented in this article has been published \cite{Boutellier2016} recently, and in this article the MoC is complemented with design rules that enable decidability analysis.

Based on the MoC, the article describes a novel Linux-based Open Source\footnote{https://gitlab.com/jboutell/Prune} programming framework PRUNE (PSM Runtime Environment) targeted for high-performance signal processing applications. PRUNE 
\begin{enumerate}
	\item implements application consistency analysis,
    \item provides an efficient runtime memory- and concurrency management framework for heterogeneous platforms, 
    \item presents a compile-time translator that allows importing programs from previous similar run-time frameworks.
\end{enumerate}
Out of these, items 1) and 3) are novel compared to \cite{Boutellier2016}.

To the best knowledge of the authors, the PRUNE dataflow framework is the first to simultaneously provide a) flexibility for describing signal processing applications with data-dependent token rates, b) a decidable Model of Computation, and c) experimental results that demonstrate high performance.

\section{Background}
\label{sec:background}

In the dataflow abstraction \cite{Lee87}, an application is described as a \textit{graph} that consists of \textit{actors} (nodes) and communication \textit{channels} (edges). Actors perform computations on data that is quantized into \textit{tokens}. Actors acquire tokens from
their input \textit{ports} and produce computation results to their output
ports. Token communication between actors is handled by order-preserving FIFO (First-In-First-Out) channels that are attached to actor ports. A dataflow actor performs
a computation by \textit{firing}, which can include consuming tokens from input
ports, and producing tokens to the actor output
ports. A central feature of the dataflow abstraction
is that computations are triggered by the availability of data, in contrast to,
for example, time-triggered abstractions \cite{Henzinger01}.

In the literature, a wide variety of dataflow Models of Computation (MoC) has been presented. One of the most important factors that differentiates a dataflow MoC from another concerns the token communication rates ({\em dataflow rates}) --- that is, the rates at which an actor reads from or writes to the channels that are connected to it. In this sense, the most restricted dataflow MoC is \textit{homogeneous synchronous dataflow} (HSDF) \cite{Lee87}, where for each actor the token rate of each input port and each output port must be exactly one.
\textit{Synchronous dataflow} (SDF) \cite{Lee87} is more expressive as it allows token rates larger than one, as is \textit{cyclo-static dataflow} (CSDF) \cite{Bilsen96}, which goes beyond SDF by allowing tokens rates to vary in repetitive cycles.

The aforementioned MoCs (HSDF, SDF, CSDF) are restricted in the sense that they
disallow \textit{data dependent} changes to the token rates, which is a
required feature as, for example, video decoders \cite{Mattavelli10} and
Software Defined Radio applications \cite{Berg08} introduce behavior that
cannot be captured by static token rates. To achieve this, \textit{dynamic
dataflow} MoCs are required. Examples of dynamic dataflow MoCs are
\textit{Boolean dataflow} (BDF) \cite{Buck93}, \textit{enable-invoke dataflow}
(EIDF) \cite{Plishker08} and \textit{dataflow process networks} (DPN)
\cite{Lee95}. Dynamic dataflow MoCs allow port token rates to change based on values of input tokens. Some formulations \cite{Lee09, Tretter15} also allow interpreting Kahn process networks (KPN) \cite{Kahn74} as a kind of a dynamic dataflow MoC.

The BDF MoC is related in some ways to that of PRUNE, however PRUNE and its design rules impose some additional restrictions that make it more analyzable than BDF. For example, PRUNE requires that (descriptions of SWITCH and SELECT can be found in \cite{Buck93})
\begin{enumerate}
    \item each SWITCH type actor needs to have a corresponding SELECT actor, unlike BDF; 
	\item data streams that control a pair of SWITCH and SELECT actors need to be identical, which is not the case in BDF; 
    \item the input token rate and output token rate on each end of a FIFO need to be identical. 
\end{enumerate}
These three differences are related to minimal examples presented in \cite{Buck93} that break either the \textit{strong consistency} or \textit{bounded memory} assumptions in BDF and preclude decidability.

\subsection{Analyzability of Dynamic Dataflow MoCs}

The disadvantage of dynamic dataflow MoCs is their limited analyzability. A
promising approach to provide analyzability and structure to dynamic behavior
is parameterization, which restricts the allowed dynamic application behavior
to a certain extent. An example of a well-known MoC belonging to this class is
Parameterized Synchronous Dataflow (PSDF) \cite{bhattacharya2001parameterized}. 

Parameterization, however, does not imply any guarantee on decidability, which means compile-time application analyzability for deadlock freedom and bounded memory. A dataflow MoC that is both decidable and dynamic is Scenario Aware Dataflow (SADF) \cite{stuijk2011scenario}, where each operation mode (scenario) of an application is expressed as a separate SDF graph. Unfortunately, in some situations, such as when there is an arbitrary
sample rate change in a signal processing application, the use of the SADF
model yields an unwieldy number of scenarios.
\textit{Parameterized sets of modes} ({\em PSMs})~\cite{Lin2015} is a 
modeling approach that addresses this shortcoming via \textit{mode sets} that enable compact management of changes both in dataflow graph topology and sample rates.
Intuitively, a PSM is a modeling abstraction that groups together a collection of related operating modes, where one or more parameters are used to select a unique mode from the collection at compile-time or run time. For details on the PSM modeling approach, we refer the reader to~\cite{Lin2015}.

A different approach for providing decidability to dynamic dataflow is presented by Gao et al. in \textit{well-behaved dataflow} (WBDF) \cite{gao1992well}: the use of dynamic actors is restricted to pre-defined actor patterns that have been shown to provide decidability. Advantages of this approach are that it can be adapted to various dataflow MoCs, and extended with new patterns, as needed.

The PRUNE MoC presented in Section~\ref{sec:proposed_model} is accompanied with \textit{design rules} (Section~\ref{sec:rules}) in the spirit of WBDF, in order to formulate necessary conditions for guaranteeing decidability, while still maintaining
support for dynamic application behavior. The PRUNE MoC also draws from the PSM concepts for compact representation of token rate changes. PRUNE applications undergo a consistency analysis at compile time (Section~\ref{sec:designtime}),
and are executed under an efficient runtime system (Section~\ref{sec:runtime}) that targets heterogeneous platforms.

\begin{table}
\caption{Comparison to related dataflow models and languages.}
\label{table:previous}
\begin{tabular}{p{2.5cm}p{1.2cm}p{1.2cm}p{2.2cm}}
\hline\noalign{\smallskip}
Work & Decidable & Dynamic & High-performance \\
\hline 
SADF \cite{stuijk2011scenario} & + & + & ? \\
\hline 
BDF \cite{Buck93} & - & + & ? \\
\hline 
DAL \cite{Schor13} & - & +(-) & + \\
\hline 
RVC-CAL/Orcc \cite{Yviquel13} & - & + & + \\
\hline 
StreamIt \cite{Hyunh14} & + & - & + \\
\hline 
PRUNE & + & + & + \\
\hline 
\noalign{\smallskip}
\end{tabular}
\end{table}

\subsection{Related Programming Frameworks}
\label{sec:related}

A number of programming frameworks that target heterogeneous platforms have emerged in the last several years. The frameworks described in \cite{Boulos14}, \cite{Sbirlea12} and \cite{Gautier13} represent \textit{task-based} programming approaches, where tasks are spawned, executed and finished, and their interdependencies are expressed as a directed acyclic graph. The proposed approach, in contrast, is based on actors that are created once at initialization and run as independent entities, communicating with each other until termination of the application.

Concerning actor based programming frameworks, a recent article \cite{Hyunh14} presents a framework that enables deploying applications written in the StreamIt language \cite{Thies02} to GPUs. Compared to this work, the significant difference is that the StreamIt language heeds the SDF MoC, which does not allow run-time changes in token rates. The same token rate restriction applies to two recent works \cite{Lund15, Boutellier15S} that discuss deployment of RVC-CAL dataflow programs to heterogeneous architectures.

The DAL framework \cite{Schor12} is based on Kahn process networks and also has an extension \cite{Schor13} for targeting heterogeneous systems with OpenCL enabled devices. In terms of OpenCL / GPU acceleration, this framework is limited to the SDF MoC, which disallows dynamic token rates.

Representative previous works are compared to the proposed PRUNE framework in Table~\ref{table:previous} in terms of decidability, support for dynamic token rates, and experimentally demonstrated high performance on heterogeneous platforms.

\section{Proposed Model of Computation}
\label{sec:proposed_model}

A common feature in signal processing oriented MoCs is the use
of semantic restrictions that a) enhance the potential
for application analysis and optimization, while b) being compatible
with specialized classes of signal processing applications. 
The PRUNE MoC has been designed for capturing the behavior of high-performance signal processing applications that can be viewed as
having configurable-topology, symmetric-rate dataflow behavior.

Here by {\em symmetric-rate dataflow}, we mean a restricted
form of SDF in which the token production rate is equal to the consumption rate
on every FIFO channel.  
However, the PRUNE MoC is
significantly more flexible than SDF in that the
connections between actors (graph topology) can be changed at run-time.
At the same time, the design rules (Section~\ref{sec:rules}) that are imposed
on the construction of PRUNE graphs ensure that important
decidability properties are maintained,
including analyzability for bounded memory and deadlock-free operation.

\subsection{Connections of a PRUNE Graph }

In the PRUNE MoC, an application is described as a graph $G =
(A, F)$, where $A$ is a set of actors and $F$ is a set of FIFO communication
channels that interconnect the actors. Each actor $a \in A$ may have zero or
more input ports and zero or more output ports. If an actor
$a$ has no input ports it is called a \textit{source actor}, and if it has
no output ports it is called a \textit{sink actor}. If an actor $a$
contains port $p$, we say that $a$ is the {\em parent} of $p$,
denoted $\mxparent(p)$. When needed for clarity, we denote with a superscript +/- the output/input direction of a port, and with a subscript number/letter the index/parent of a port. For example, $p_{a1}^+$ denotes output port \#1 of actor $a$.

Each FIFO $f \in F$ is connected to an output port $p^+$ of some actor
$\mxparent(p^+)$, and to an input port $p^-$ of some actor
$\mxparent(p^-)$.  The ports $p^+$ and $p^-$ are referred to, respectively, as
the {\em source port} and {\em sink port} of $f$. We say that the ports $p^-$ and $p^+$ are {\em connected} when
they are source and sink ports of the same FIFO --- that is, when $\mxfifo(p^-)$ = $\mxfifo(p^+)$.

A given output port $p^{+}$ can be connected to multiple FIFOs. However, each FIFO
has a unique source port and sink port, and each input port $p^{-}$ has a unique FIFO
that connects to it. When a source port $p^+$ is connected to multiple FIFOs, and
$\mxparent(p^+)$ writes a token through $p^+$, the token is written (i.e.,
broadcasted) into all of the FIFOs connected to $p^+$. 

Each port $p$ has a type that is either a \textit{control input port} (``control port''), a {\em static regular
port} ({\em SRP}) or a {\em dynamic regular port} ({\em DRP}). 
SRPs have a single, fixed, positive token consumption rate (for input ports), or token production rate (for output ports). DRPs, in contrast, have two fixed token rates that are referred to as the
{\em active token rate} (\mxatr) of $p$, and denoted as $\mxatr(p)$, and the {\em inactive token rate} (\mxitr), which is always zero. The consumption rate of a control port is always equal to unity.

A distinguishing aspect of the PRUNE MoC
is that each FIFO has a single, positive-integer token rate, denoted
by $\mxfiforate(f)$, that is associated with it.
In conjunction with this association of token
rates with FIFOs, the PRUNE MoC imposes the restriction
for each port $p$, $\mxatr(p) = \mxfiforate(\mxfifo(p))$.
In other words, it is a semantic error to have
a port that is connected to a FIFO if there is a mismatch between 
the $\mxatr$ of the port and the token rate of the FIFO.

Similar to KPN \cite{Kahn74}, in the PRUNE MoC, blocking reads and blocking writes are assumed for all actors. Under blocking a read, the execution of an actor stalls if it has an input port $p$ such that the number of tokens in $\mathit{fifo}(p)$ is less than the number of tokens in $\mxfiforate(p)$. The same holds for output FIFOs and their free token slots.

\subsection{Types of Actors}

Computations related to a PRUNE application are performed in firings of actors. A firing $\phi$ of actor $a$ consumes tokens from the input ports (FIFOs) of $a$ and produces tokens to the output ports (FIFOs) of $a$. In the PRUNE MoC, each actor has a type that is either \textit{static processing actor}, \textit{dynamic actor}, or \textit{configuration actor}. The allowed port types and their associated firing behavior depend on the type of the actor. 

A dynamic actor $x$ has at least one DRP, any number of SRPs, and a unique control input port, denoted $\mxcontrol(x)$.
When dynamic actor $x$ performs a firing $\phi$, a \textit{control token} is consumed from $\mxcontrol(x)$. The control token sets the token rate for each DRP $p$ of $x$ to $\mxatr(p)$ or $\mxitr(p)$ for the duration of $\phi$.
By definition, the token rate of each SRP of $x$ remains at its fixed positive token rate regardless of the value of the corresponding control token. 

\begin{figure}
\centering
\includegraphics[width=0.88\linewidth]{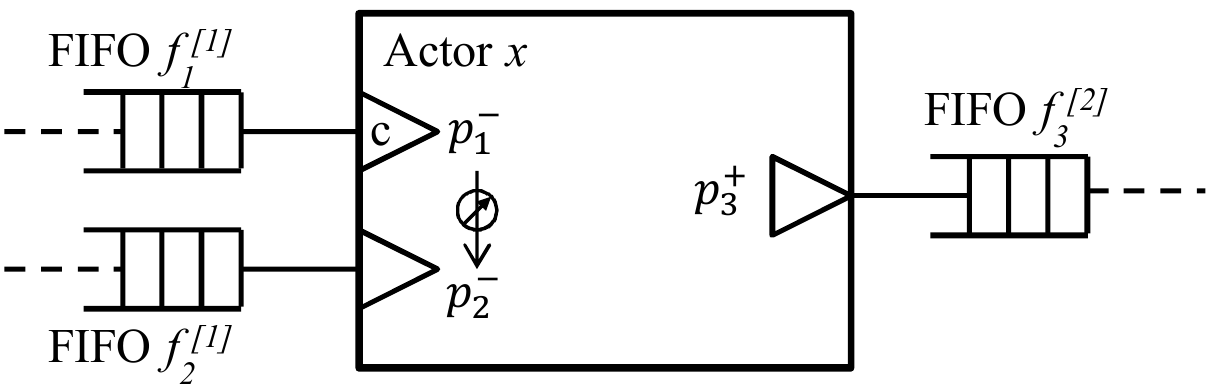}
\caption{An illustration of a dynamic actor in the PRUNE MoC.}
\label{fig:actor}
\end{figure}

Fig.~\ref{fig:actor} depicts an example of a dynamic actor. This actor,
denoted $x$, is connected to three FIFO channels, $f_1$, $f_2$ and $f_3$,
through its ports $p_1^{-}$, $p_2^{-}$ and $p_3^{+}$. FIFOs
$f_1$, $f_2$ and $f_3$ have token rates of 1, 1, and 2 (shown in brackets), respectively. The
annotation of port $p_1^{-}$ with ``c'' indicates that this is the
control port of the dynamic actor. Port $p_2^{-}$ is a DRP, while
$p_3^{+}$ is an SRP. Values of the tokens consumed from the control port
set the token rate of input port $p_2^{-}$ to either 0 or 1; in other
words, $\mxatr(p_2^{-}) = 1$ and $\mxitr(p_2^{-}) = 0$.

A configuration actor has one or more {\em control output ports}. A control output port of a configuration actor
must be an SRP, have a token production rate of unity, and
be connected to the control input port of a dynamic actor.
Thus, the control tokens consumed by control input ports are always
produced by configuration actors: if $p_x^-$ is a control input port of a dynamic actor $x =
\mxparent(p_x^-)$, then there is a unique control output port $p_q^+$ of a
configuration actor $q = \mxparent(p_q^+)$ such that $p_x^-$ is connected to
$p_q^+$. In addition to its one or more control output ports, a configuration actor has zero or more {\em data ports} of type SRP. A data port can be either an input port or output port.

Finally, all ports of a static processing actor $a$ are of the type SRP and hence active during all firings of $a$. Thus, for all firings $\phi$ of a static processing actor $a$, and all ports $p$ of $a$, $\mxtokrate(p, \phi) = \mxatr(p)$, where $\mxtokrate(p, \phi)$
denotes the token rate of a port $p$ during a firing $\phi$ of $\mxparent(p)$.

\subsection{Control and Firing of a Dynamic Actor}

Since each control output port $p_q^+$ of a configuration actor $q$ is required to connect to a control input port $p_x^-$ of a dynamic actor $x$, and in turn $p_x^-$ sets the token rate of each DRP of $x$, we say that the port $p_q^+$ {\em controls} the DRPs of $x$. 
This control relationship can be represented as part of a data structure called the {\em control table}. The control table for a  PRUNE graph or subgraph $G$ is a matrix whose rows are indexed by control output ports, and columns are indexed by DRPs. The size of the matrix is $h \times w$, where $h$ and $w$ are the number of control output ports and DRPs, respectively, in $G$. Fig.~\ref{fig:control_table} provides an example of a control table.

When a dynamic actor $x$ fires, it first consumes a 
control token from its control port. A control token in turn encapsulates a
{\em control value} $\bar{v}$, which is a vector $\bar{v}[1], \bar{v}[2], \ldots, \bar{v}[K]$ of
Boolean elements, and $K$ is the number of DRPs in $x$. In other words, there is one element in $\bar{v}$ for each DRP of $x$ (For reasons of clarity, here $K$ is assumed to equal the DRP count of $x$. A more general formulation is presented in Section~\ref{sec:designtime}). 
A control value $\bar{v}$ produced on a port $j$ is related to the control table in that each entry $T[j][i]$ of the control table indicates the element index in $\bar{v}$ that is used to configure the DRP $i$.
For example, $T[j][i] = 3$ means that the third element of each control token produced on port $j$ is used to control DRP $i$. If port $j$ is not used to configure DRP $i$, then $T[j][i] = 0$.

Based on the control value of the consumed control token, the
token rate of each DRP $p$ of $x$ is
fixed to either 0 or $\mxatr(p)$ \textit{for the duration of the current firing} $\phi$. More specifically, the token rate on DRP $p$
for firing $\phi$ is determined by 
\begin{equation} 
\label{eq:dataflow}
\mxtokrate(p, \phi) = \mxbtoi(\bar{v}_\phi[p]) \times \mxatr(p),
\end{equation}
\noindent where $\bar{v}_\phi$ denotes the control
value consumed in firing $\phi$, and $\mxbtoi$ 
(Boolean to integer) represents 
a simple function that maps Boolean values to integer values: that is,
$\mxbtoi(\mxfalse) = 0$, and $\mxbtoi(\mxtrue) = 1$. 

The computation associated with a given firing $\phi$ must adhere to the
dynamically-adjusted dataflow constraints imposed by
Equation~\ref{eq:dataflow}.  
Implementation of actor firing functions in PRUNE is discussed in Section~\ref{ssec:actors}.

\begin{figure}
\centering
\includegraphics[width=\linewidth]{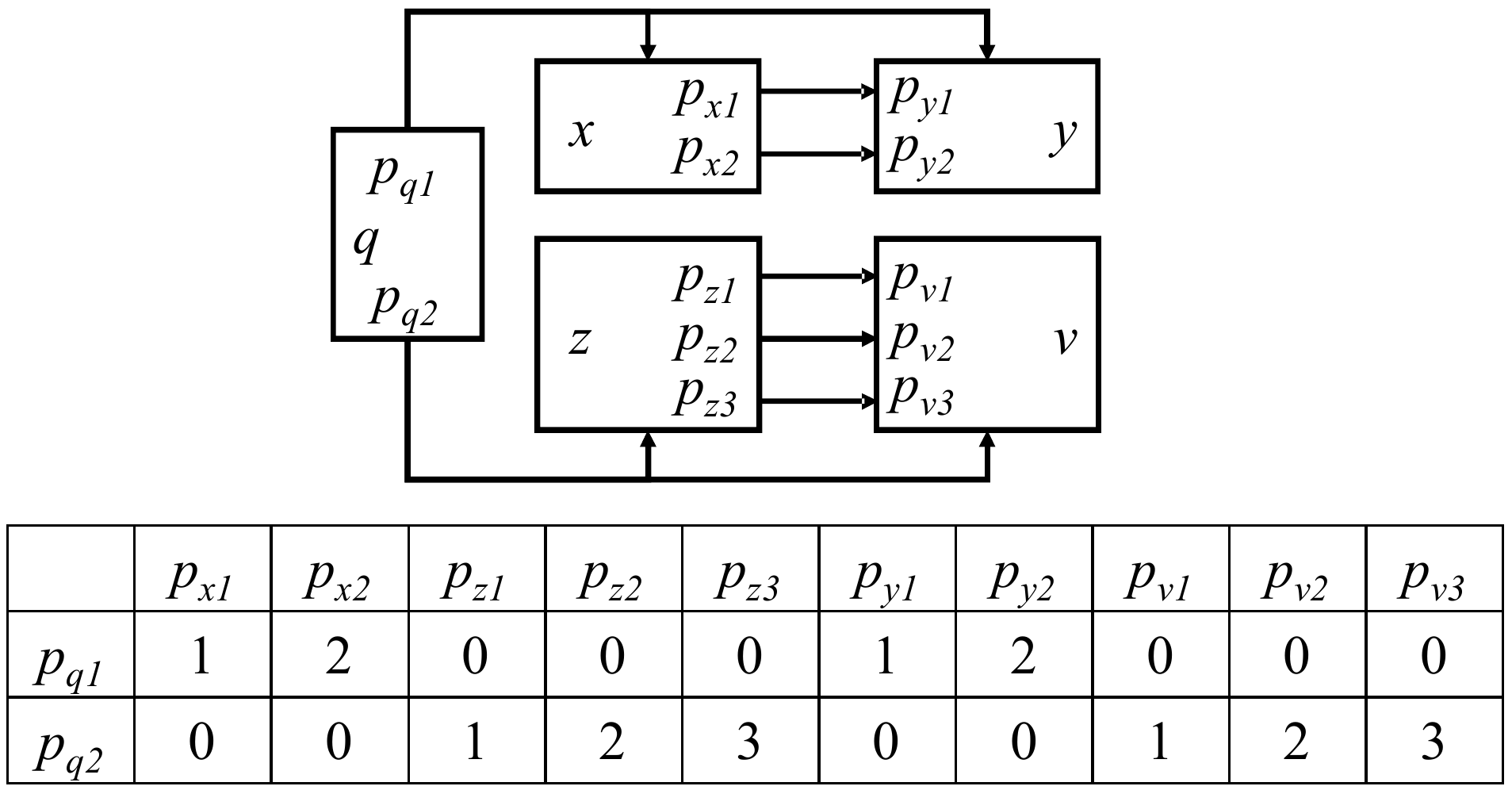}
\caption{A PRUNE graph and its control table. Configuration actor $q$ controls the DRPs of dynamic actors $x$, $y$, $z$ and $v$.}
\label{fig:control_table}
\end{figure}

\subsection{Token Delays}

Each FIFO channel $f \in F$ has a non-negative integer {\em delay} associated
with it, which specifies the number of initial tokens that are placed in the
channel at system setup time (before the graph is executed). Such delays can be
used, for example, to implement the $z^{-1}$ operator in signal processing
(e.g., see~\cite{bhat2013x1}). If $p^+$ and $p^-$ are two ports that are connected
to a common FIFO $f$, then (with a minor abuse of notation) we denote the
delay associated with $f$ by $\mxdelay(p^+, p^-)$  or by $\mxdelay(f)$.
It is of high importance to notice that the presence of delays on FIFOs combined with PRUNE's \textit{symmetric-rate} dataflow behavior can lead to unaligned FIFO accesses. This is discussed in Section~\ref{ssec:channels}.

\section{Design Rules}
\label{sec:rules}

Arbitrary connections of dynamic actors can lead to inconsistent dataflow behavior \cite{Buck93}.
Such inconsistencies can lead to deadlock or unbounded 
accumulation of tokens within FIFOs, which are problematic
when a signal processing system must operate on very large, possibly
unbounded streams of data~\cite{bhat2013x1}.

PRUNE imposes a small set of concrete design rules
to ensure that dynamic activation of ports in a given graph 
is performed in a manner that maintains matched patterns
of token production and consumption across each FIFO in a PRUNE graph.
Thus, tokens that need to be consumed are ensured to have corresponding
producers (producing actors), and similarly, production of
tokens is matched with a corresponding ``demand'' to consume the produced data.
The design rules therefore ensure consistency
in the dynamic dataflow behavior of PRUNE graphs.

In the remainder of this section, we formulate the design rules of PRUNE.
The precise formulations provided here in terms of fundamental dataflow
and graph-theoretic concepts allow the rules to be checked automatically
within design tools. Indeed, in our prototype analysis tool for PRUNE, 
which we report on in  Section~\ref{sec:designtime}, the
design rules are checked automatically to aid the designer in iteratively
refining a design as needed until all design rules are satisfied.

As necessary background we first
review some graph theoretic concepts in the context of PRUNE graphs. We say that
two PRUNE actors $a$ and $b$ are {\em adjacent} if there is
a FIFO that connects a port of $a$ to a port of $b$.
A {\em chain} in a PRUNE graph is a non-empty sequence
$S = (a_1, a_2, \ldots , a_N)$ of actors in the graph such that
for each $i = 1, 2, \ldots, (N - 1)$, $a_i$ and $a_{i+1}$ are
adjacent. We say that the chain $S$ {\em connects} $a_1$ and $a_N$.
The chain $S$ is a {\em simple chain} if
all of the $a_i$'s are distinct. Given two actors $x$ and $y$, an actor $z$ $\notin \{x, y\}$ is said to \textit{connect} actors $x$ and $y$ if there is a chain $S$ that connects $x$ and $y$ such that $S$ contains $z$.

Now suppose that $p_x$ and $p_y$ are distinct ports within two actors $x$ and
$y$, respectively, of a PRUNE graph $G$.  We say that $p_x$ and $p_y$ are {\em
linked ports} if (a) $\mxfifo(p_x) = \mxfifo(p_y)$  or (b) there is a simple
chain $(x, a_1, a_2, \ldots , a_N, y)$ of actors, where $p_x$ is connected to a
port of $a_1$, and $p_y$ is connected to a port of $a_N$.  The sequence of
$a_i$s in (b) is referred to a {\em connecting subchain} associated with the
linked ports $\{p_x, p_y\}$. Note that for the same linked ports $\{p_x, p_y\}$,
there can be multiple connecting subchains. If $p_x$ and $p_y$ are linked
ports, and they are both DRPs, then we say that they are {\em linked DRPs}.

The {\bf first design rule}, called
the {\em linked port control rule}, 
is that for each pair $\{p_x, p_y\}$ of linked DRPs, the ports must be controlled by the same control output port $p_q$, and by the same element of the associated control token --- that is, $T[p_q][p_x] = T[p_q][p_y] > 0$.

The {\bf second design rule}, called the {\em balanced delay rule}, states that if a control output
port $p_q$ controls DRPs $\mxcontrol(x)$ and $\mxcontrol(y)$, then 
$\mxdelay(p_q, \mxcontrol(x)) = \mxdelay(p_q, \mxcontrol(y))$.
In other words, the control input ports of $x$ and $y$
must be connected to $p_q$ with the same delay.

The {\bf third design rule}, called the {\em connecting subchain rule}, states that if $x$ and $y$ are dynamic actors, $\{p_x, p_y\}$ are linked DRPs with $\mxparent(p_x) = x$ and $\mxparent(p_y) = y$, $S = (a_1, a_2, \ldots , a_N)$ is a connecting subchain associated with $\{p_x, p_y\}$, then (1) actor $a_i \in S$ must be a static processing actor, and (2) each connecting subchain, to which actor $a_i \in S$ belongs, must be associated with the two dynamic actors $x$ and $y$.

The {\bf fourth design rule}, called the {\em single-sided dynamism rule}, states that a dynamic actor may only have dynamic input ports, or dynamic output ports, but not both.

The {\bf fifth design rule}, called the {\em encapsulation rule}, states that if $x$ and $y$ are dynamic actors, $\{p_x, p_y\}$ are linked DRPs with $\mxparent(p_x) = x$ and $\mxparent(p_y) = y$, $S = (a_1, a_2, \ldots , a_N)$ is a connecting subchain associated with $\{p_x, p_y\}$, and $b \notin S$ is an actor that connects to an actor $a_i \in S$ through an SRP of $b$, then $b$ must belong to a chain that connects $x$ and $y$.

\begin{figure}
\centering
\includegraphics[width=0.89\linewidth]{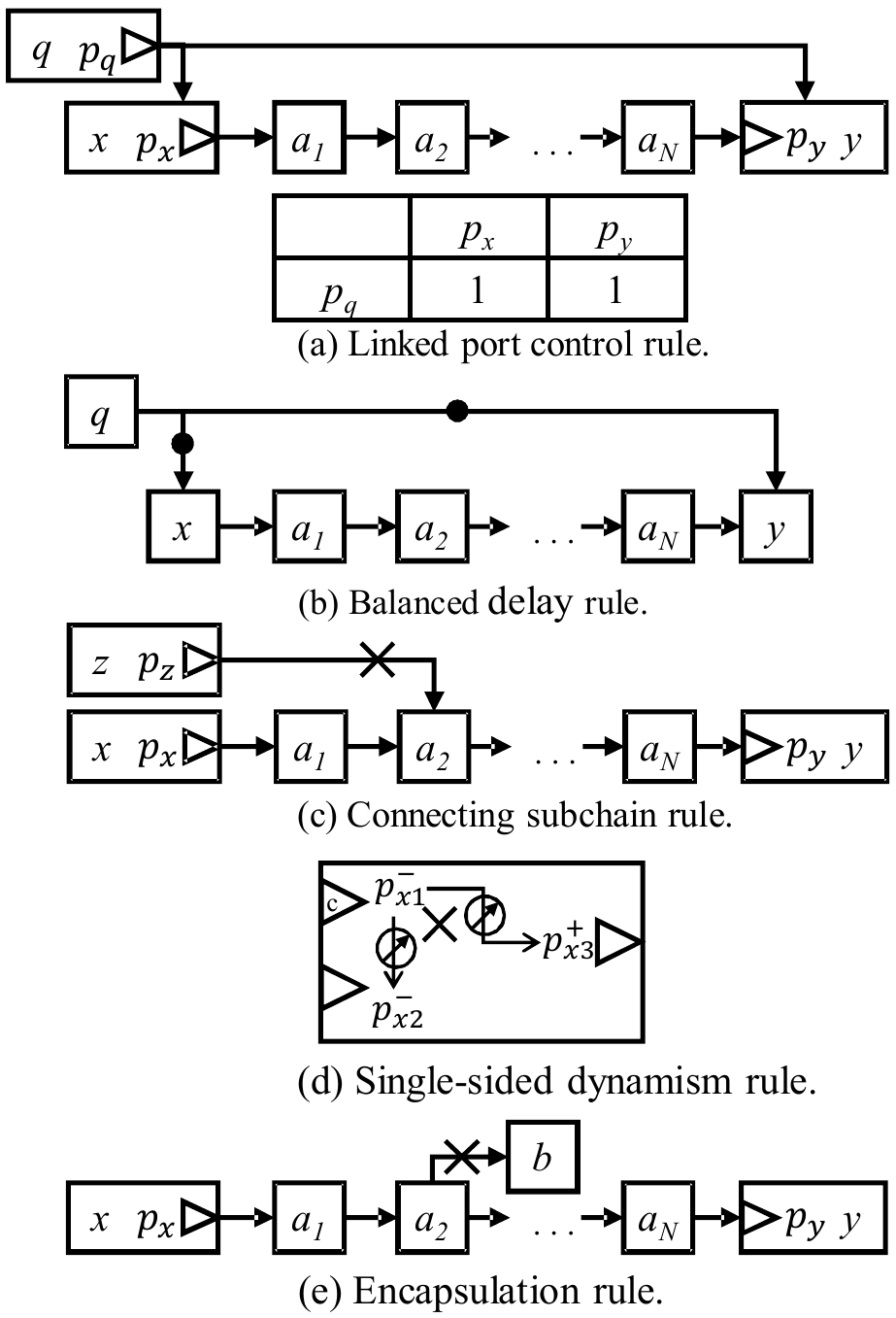}
\caption{An illustration of design rules in PRUNE.}
\label{fig:designrules}
\end{figure}

Fig.~\ref{fig:designrules} provides illustrations of the five design rules
in PRUNE. The control table shown in the lower part of
Fig.~\ref{fig:designrules}(a) represents relationships between the control
output port $p_q$ and two DRPs --- DRP $p_x$ of actor $x$, and DRP $p_y$ of
actor $y$. The control table shows that since DRPs $p_x$ and $p_y$ are controlled by the same control output port $p_q$
and the same element of the associated control token, the linked
port control rule is satisfied. Fig.~\ref{fig:designrules}(b) shows a graph that satisfies the balanced delay rule, whereas Fig.~\ref{fig:designrules}(c) shows an example that violates the connecting subchain rule: actor $a_2$ of the connecting subchain $(a_1, a_2, \ldots , a_N)$ associated with the dynamic actors $x$ and $y$ also belongs to another connecting subchain, associated with dynamic actors $z$ and $y$. In Fig.~\ref{fig:designrules}(d), the dynamic actor $x$ violates the single-sided dynamism rule, as it contains both input and output DRPs.
Finally, Fig.~\ref{fig:designrules}(e) depicts a violation of the encapsulation rule. Here, actor $a_2$ belongs to a chain that connects dynamic actors $x$ and $y$. Actor $a_2$ is adjacent to actor $b$, but $b$ is not part of a chain that connects $x$ and $y$.

\section{Compile Time Graph Analysis}
\label{sec:designtime}
\noindent

The aim of the proposed design rules is to ensure that deadlock freedom and bounded memory analysis of a PRUNE graph can be completed in finite time, i.e. these problems remain decidable. This section establishes this decidability result for the PRUNE model of computation.

The design rules require that DRPs and their dynamic actor parents $x$ and $y$ always appear in pairs and both of these are controlled by the same configuration actor $q$. Therefore, in a PRUNE graph $G$ we can identify zero or more {\em dynamic processing graphs} (DPGs) that each consist of a) one configuration actor $q$, b) exactly two dynamic actors, $x$ and $y$, and c) any number of chains that connect $x$ and $y$. These chains form the {\em dynamic components} (DCs) of the DPG. Given a DPG $D$, the set of DCs of $D$ is denoted $Z_c(D)$, and the pair of dynamic actors contained in $D$ is denoted $\delta(D)$.

Consider a DPG that contains dynamic actor $x$ with $K$ output DRPs $p_{xi}$ $(i = 1, 2, \ldots, K)$, and dynamic actor $y$ with $L$ input DRPs $p_{yj}$ $(j = 1, 2, \ldots, L)$, $\{x,y\}=\delta(D)$; we require that each $p_{xi}$ is a linked DRP with at least one of $p_{yj}$. Our procedure for finding the DCs $Z_c(D)$ associated with a given DPG $D$ can be expressed as follows:

1. For each linked DRP $\{p_{xi}$, $p_{yj}\}$, where $\mxfifo(p_{xi}) = \mxfifo(p_{yi})$, insert a dummy actor $d$ such that $\mxfifo(p_{xi}) = \mxfifo(p_{d}^{-})$ and $\mxfifo(p_{d}^{+}) = \mxfifo(p_{yj})$.

2. Remove $q$, $x$, $y$, and all FIFOs $\mxfifo(p_q)$, $\mxfifo(p_x)$ and $\mxfifo(p_y)$ in the DPG. This removal procedure decomposes the DPG into a set of connected components that form the DCs. Thus, $Z_c(D) = \{Z_1, Z_2, \ldots, Z_M\}$, where $M \in [1, min(K,L)]$ is an integer constant.

As an example, consider the DPG in Fig.~\ref{fig:analysis} that consists of the configuration actor $q$, the pair of dynamic actors $\delta(D) = \{x, y\}$, and actors $a_1$ through $a_4$. It can be seen that the linked DRPs of this DPG are $\{p_{x1}, p_{y1}\}$, $\{p_{x2}, p_{y1}\}$, $\{p_{x3}, p_{y2}\}$, and $\{p_{x4}, p_{y3}\}$.

The linked port-pair $\{p_{x3}, p_{y2}\}$ is connected over a one-actor chain of $S = (a_4)$, whereas the ports $\{p_{x4}, p_{y3}\}$ are linked directly, i.e., $\mxfifo(p_{x4}) = \mxfifo(p_{y3})$. Here, the DC analysis procedure inserts the dummy actor $d$. Finally, for $\{p_{x1}, p_{y1}\}$ and $\{p_{x2}, p_{y1}\}$, the adjacent actors $a_1, a_2$, and $a_3$ form one DC. Hence, in this example, the DPG consists of three DCs $Z_c(D) = \{Z_1, Z_2, Z_3\}$, where $Z_1 = \{a_1, a_2, a_3\}$, $Z_2 = \{a_4\}$, and $Z_3 = \{d\}$. The dummy actor $d$ exists only for the duration of the graph analysis and is not carried to the final implementation.

For a DPG to be \textit{valid}, we require that
a) each DC in $Z_c(D)$ is connected to at least one DRP of $x$ and to at least one DRP of $y$ ($\{x,y\}=\delta(D)$), and
b) there must be exactly one Boolean element in $\bar{v}$ for each DC in $Z_c(D)$. Thus, the number of elements in $\bar{v}$ must equal $M$.

Expressed through binary relations of sets, the relation between DCs of $Z_c(D)$ and the elements of $\bar{v}$ is required to be a \textit{bijection}. In contrast, the relation $\mxdrps(x) \rightarrow Z_c(D)$ (and $\mxdrps(y) \rightarrow Z_c(D)$) is surjective, but not necessarily injective. Here, the set of all DRPs of $x$ is denoted as $\mxdrps(x)$. In other words, the DRPs associated with a specific $Z_k$ must all simultaneously be configured with their respective $\mxatr$s, or they all must be configured with their $\mxitr$s.

Consequently, in a valid DPG $D$, each control token (which encapsulates $\bar{v}$) produced by $q$ sets each DRP $p_x$ of $x$ and each DRP $p_y$ of $y$ either to its $\mxatr$ or its $\mxitr$, where $\{x, y\}$ = $\delta(D)$. A control token effectively sets each DC in $Z_c(D)$ as \textit{active} or \textit{inactive}. In this context, we say that the DC $Z_k \in Z_c(D)$ is active if it will be provided with tokens through DRPs of $x$, and the actors within $Z_k$ will consequently fire producing tokens that are consumed by actor $y$. Inactiveness of $Z_k$, in contrast, means that DRPs of $x$ will not provide tokens to $Z_k$, and consequently no actor in $Z_k$ will fire, and actor $y$ does not demand tokens from $Z_k$.

Returning to Fig.~\ref{fig:analysis}: due to Design Rules 1 and 2, for a valid DPG, for example the linked DRP pairs $\{p_{x1}, p_{y1}\}$ and $\{p_{x2}, p_{y1}\}$ are all simultaneously set either to their $\mxatr$ or to their $\mxitr$; a deadlock would follow if any of the three ports would be set to its $\mxitr$, while the other would be set to its $\mxatr$ (and vice-versa).

\begin{figure}
\centering
\includegraphics[width=\linewidth]{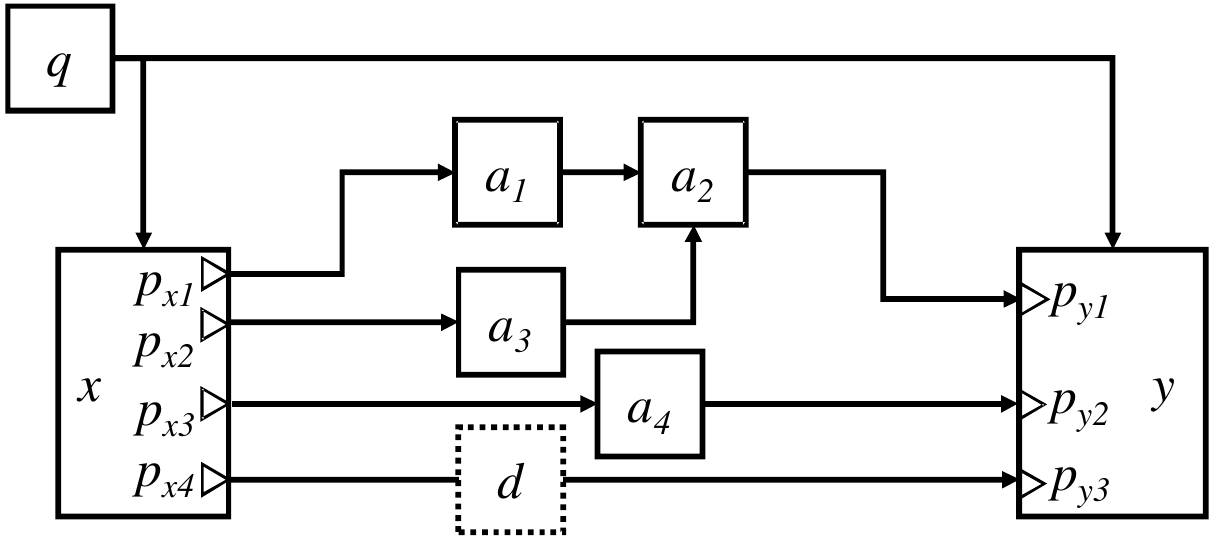}
\caption{PRUNE compile time analysis example for one DPG. Connections to the complete PRUNE graph $G$ are not shown.}
\label{fig:analysis}
\end{figure}


Above, the necessary background information has been given for our discussion on the decidability of PRUNE graphs. As the design rules of Section~\ref{sec:rules} restrict dynamic actors to exist within DPGs, actors outside DPGs have fixed data rates and have dataflow relationships with DPGs that are simple, and can readily be validated using standard SDF techniques. Thus, in the remainder of this section, the decidability discussion concentrates on DPGs. 

In general, a PRUNE graph $G$ may contain multiple DPGs. However, our design rules require the DPGs of $G$ to be independent of each other. Since the number of distinct DPGs is finite, our proof of decidability can be reduced to proving that consistency analysis for a single DPGs is decidable.

\begin{defi}[Consistency]
A PRUNE graph is consistent if it can be scheduled with guarantees of bounded memory requirements and deadlock-free operation, regardless of what inputs are applied.
\end{defi}

\begin{theorem}
The consistency analysis of a PRUNE graph is decidable.
\end{theorem}

\begin{proof}[Proof]
Let $Z_c(D) = {Z_1, Z_2, \ldots, Z_M}$ be the set of DCs of DPG $D$. Since Design Rule 3 requires that all actors in DCs are static processing actors, there is a finite number of different SDF graphs $Z_1, Z_2, \ldots, Z_M$ that can be active during execution of the DPG $D$.

Since each $Z_k, k \in [1,M]$ is an SDF graph, it is decidable to determine whether or not the graph is consistent \cite{Lee87}. If any of the $Z_k$'s is not consistent, then deadlock-freedom and bounded memory scheduling for the DPG $D$ cannot be guaranteed, and therefore $D$ is inconsistent. On the other hand, if all $Z_k$'s are consistent, then there exists a valid, periodic schedule $P(Z_k)$ for each $Z_k$ \cite{Lee87}. $P(Z_k)$ provides a schedule for $\mathit{actors}(Z_k)$, where $\mathit{actors}(X)$ represents the set of actors that are contained in a given DC $X$.

For each FIFO $f$ connected to an actor $a \in \mathit{actors}(Z_k)$, there is a buffer bound $B_k(f)$ which gives the maximum number of tokens on $f$ during an execution of $P(Z_k)$. The existence of this buffer bound follows from the properties of consistent SDF graphs~\cite{Lee87}. There is then a finite maximum $\beta(f) = max(B_k(f) \mid Z_k \in Z_c(D))$. 

Any execution of the DPG $D$ can be carried out by a sequence of schedules $\Omega = (O_1, O_2, \ldots$) where for each $O_k$, there is an $H_k \in Z_c(D)$ such that $O_k = P(H_k)$. $H_k$ can be viewed as the $k$th active graph during execution of the enclosing DPG.

Since each $Z_k$ is assumed to be consistent and have a valid, periodic schedule $P(Z_k)$, $O_k$ produces no net change in the token populations of the buffers between $\mathit{actors}(Z_k)$, and consequently, the number of tokens on a FIFO $f$ during an execution of $O_k$ is bounded by $B_k(f)$. It follows that the number of tokens on $f$ during execution of $\Omega$ is bounded by $\beta(f)$.

In summary, the consistency of the DPG $D$ can be determined by analyzing the consistency of the elements of $Z_c(D)$. Since the elements of $Z_c(D)$ can be analyzed in finite time (due to the decidability of SDF and the fact that $Z_c(D)$ has finite cardinality), it follows that consistency analysis for DPG $D$ is decidable. 
\end{proof}


The design rules presented in Section~\ref{sec:rules} and the analysis presented in this section allow construction and verification of DPGs within a PRUNE application graph $G$. A DPG can be seen as a generalization of the \textit{conditional schema} of \cite{gao1992well} in terms of the number and topology of conditional branches, and the fact that the branches are not mutually exclusive. Following the same logic, it is possible to formulate useful design rules for other kinds of dynamic constructs on top of the PRUNE MoC, which is a useful direction for future work.

\section{The PRUNE Framework}
\label{sec:runtime}

The previously described design rules and compile time graph analysis have been implemented to the \textit{PRUNE compiler and analyzer} as shown in Fig.~\ref{fig:framework}. The PRUNE compiler takes three types of input files: the application graph, the platform graph, and the actor-to-platform mapping.

The formats of the application graph and the platform graph have directly been adopted from the DAL framework \cite{Schor12}, but some extensions have been introduced to provide support for execution of dynamic actors on GPUs, multi-dimensional OpenCL workloads and static data for OpenCL actors.

The PRUNE compiler and analyzer transforms the XML input files into an internal representation, which is suitable for performing graph analysis and verification. If the sanity checks and the compile-time analysis (Section~\ref{sec:designtime}) for the input files pass, the compiler outputs the main program file of the application, which takes care of initializing and terminating OpenCL device access, memory allocations, actors and FIFOs.

\begin{figure}
\centering
\includegraphics[width=\linewidth]{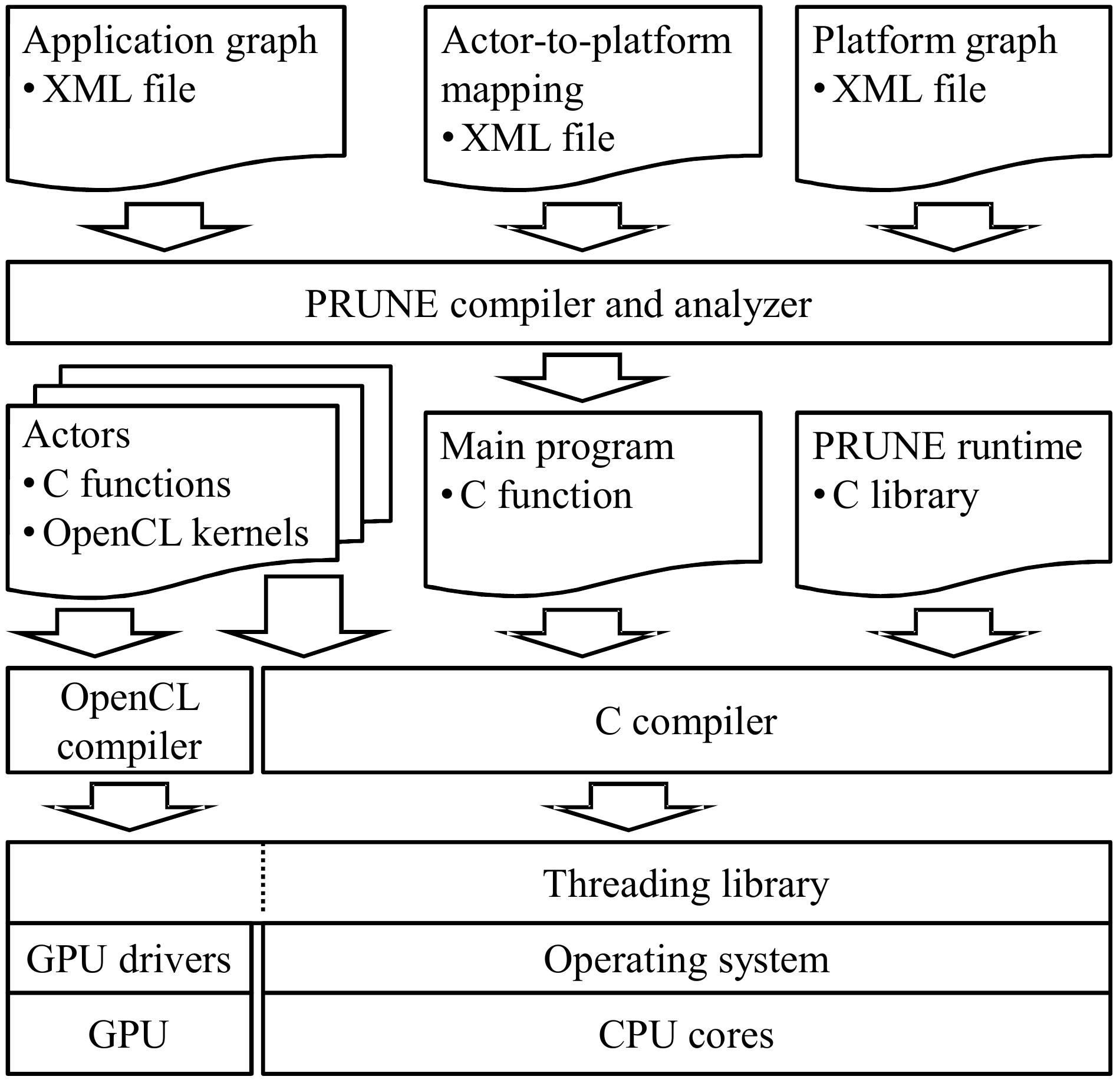}
\caption{Overview of the PRUNE framework.}
\label{fig:framework}
\end{figure}


After the PRUNE compiler and analyzer has successfully produced the main C file of the application, the application is ready to be compiled with the target platform specific C and OpenCL compilers.

Besides the main C file, this compilation step requires the functional description of each actor. The PRUNE actor API (application programming interface) closely follows the DAL API except for GPU-mapped actors that are described in OpenCL (in DAL a small translator converts appropriately formatted C actors to their OpenCL equivalents). The PRUNE actor API essentially provides functions for inter-actor communication, such as \texttt{fifoWriteStart}, \texttt{fifoWriteEnd}, etc.

Finally, compilation with the target-specific C compiler requires the PRUNE run-time library, which contains application independent actor wrappers, FIFO implementations and OpenCL support. The following detailed description of the PRUNE runtime framework contains some extension compared to our preliminary work \cite{Boutellier2016}, where it was first presented. For example, Equation~\ref{eq:fifosize} has been generalized to support token delays $> 1$.

\subsection{Description of Actors}
\label{ssec:actors}

In the PRUNE runtime framework the description of each actor consists of the mandatory \textit{fire} function, and optional \textit{init}, \textit{control}, and \textit{finish} functions. The \textit{fire} function describes the actor's behavior upon firing and comprises the reading of SRP and DRP input ports, computation and writing to SRP and DRP output ports. The optional \textit{init} and \textit{finish} functions are only executed once on application initialization and termination, and are mainly useful for source and sink actors to start and end interfacing with I/O. The \textit{control} function is only required for dynamic actors and is executed once for each firing of the actor, right before invoking the \textit{fire} function. The \textit{control} function is responsible for reading the control input port and setting the token rates of DRPs.

This formulation, where actors consist of \textit{init}, \textit{fire}, and \textit{finish} functions is identical to the DAL \cite{Schor12} framework. However, the \textit{control} function, especially required for enabling dynamic data rate actors on OpenCL / GPU devices, is specific to PRUNE. The \textit{control} function takes one control token as input and sets the token rate (to $\mxitr(p)$ or $\mxatr(p)$ as defined in Section~\ref{sec:proposed_model}) of each DRP for the duration of one firing.

\begin{figure}
\centering
\includegraphics[width=\linewidth]{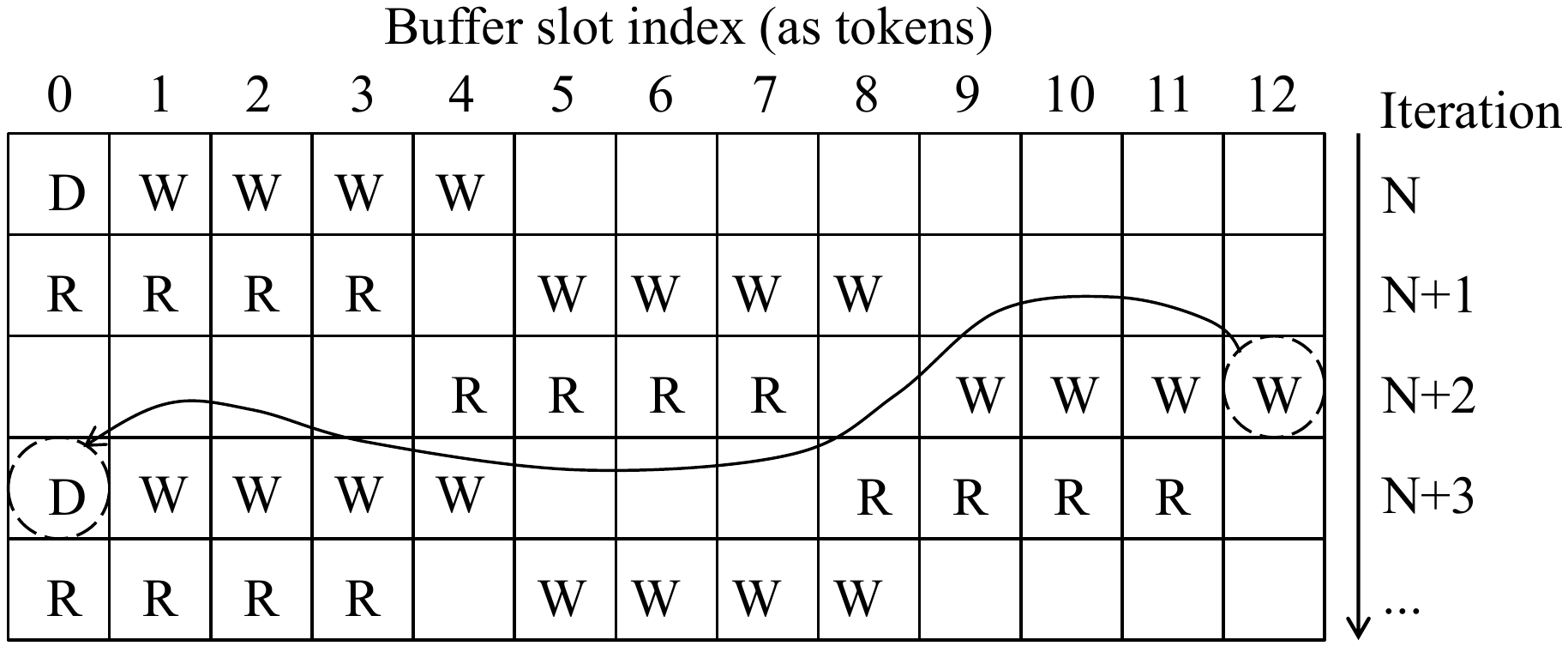}
\caption{FIFO channel token access pattern in the case of one delay token for token rate 4 and $C=3$.}
\label{fig:triplebuffer}
\end{figure}

\subsection{Communication Channels}
\label{ssec:channels}

A communication channel in the proposed framework connects an output port of an actor to an input port of another actor, heeding FIFO behavior. In contrast to other (e.g. \cite{Schor12}) programming frameworks, the capacity $\Gamma_f$ of a communication channel $f$ cannot be arbitrarily chosen by the programmer, but is exactly specified as
\begin{equation}
\label{eq:fifosize}
\Gamma_f = 
    \begin{cases}
	B * (r * C + Q), & \text{$Q$ not an integer multiple of $r$} \\
	B * max(r * C, Q), & \text{otherwise,} \\
    \end{cases}
\end{equation}
where $r = \mxfiforate(f)$, $B$ is the size (e.g. in bytes) of one token of FIFO $f$, $Q = \mxdelay(f)$, and $C$ is a compile-time constant. For example, setting $C=2$ creates a double buffer, and $C=3$ creates a triple buffer.

A regular channel (the \textit{otherwise} case in Equation~\ref{eq:fifosize}) is double or a triple buffer, which allows simultaneous reading and writing of tokens to the channel. On each write, $f$ assumes to receive $r$ tokens, and on each read the channel outputs $r$ tokens. However, for channels that contain initial tokens (delay), the channel is implemented as a slightly more complex buffer that implements a specific access pattern to enable simultaneous reads and writes to the channel. This is exemplified in Fig.~\ref{fig:triplebuffer} with $r = 4$. 

At application initialization the initial token in the channel, displayed with $D$ in Fig.~\ref{fig:triplebuffer}, resides in buffer slot 0. The first write to the channel occupies slots 1 ... 4, whereas the first read consumes tokens from slots 0 ... 3 and so forth. The third write to the channel reaches the end (slot 12) of the buffer, followed by an explicit data copy from slot 12 to slot 0, and the access pattern starts to repeat. The access pattern is repetitive and can be generalized to any token rate beyond one. 

Looking at Fig.~\ref{fig:triplebuffer}, it is evident that this solution does not minimize the memory footprint, but it was chosen as it offers 1) uncompromised throughput and 2) transparency to the application programmer. Ring buffers were considered inadequate, as OpenCL / GPUs offer the best combination of performance and ease of programming when input and output data to actors is provided as contiguous arrays. 

\subsection{Concurrency, Scheduling and Actor-to-Core Mapping}
\label{ssec:concurrency}

The proposed framework has been designed to enable maximally parallel operation. Parallelism is based on threading, such that each actor runs on an operating system (OS) thread of its own, regardless whether the actor is targeted to OpenCL / GPU devices or to one of the CPU cores. Each actor thread is created once at application startup, and is canceled after the application has terminated. Similar to the DAL framework \cite{Schor12, Schor13}, synchronization of data exchange over FIFO channels is based on \textit{mutex} locks and blocking communication: if an actor attempts to read a channel that has less tokens than the actor requires, the reading actor blocks until sufficient data is available. This enables very efficient multiprocessing, but on the other hand makes the MoC somewhat more restricted than e.g. that of DPNs \cite{Tretter15}.

As each actor is instantiated as a separate thread using the GNU/Linux \textit{pthreads} library, the scheduling of actor firings (heeding data availability) is left to the OS. If the programmer so chooses, the framework allows fixing of actors to specific CPU cores, otherwise the OS chooses the core on which the actor is executed.

It is necessary to state that alternatively to the adopted OS threading based concurrency, it would also have been possible to build concurrency and synchronization on top of OpenCL events, however this would have limited the applicability to platforms where both the CPU cores and GPUs have OpenCL drivers. The adopted OS threading based solution, however, is beneficial due to its backwards compatibility: with this solution it is possible to jointly synchronize and run also non-OpenCL compatible CPU cores with GPUs.

\begin{table*}
\caption{Platforms used for experiments.}
\label{table:platforms}
\begin{tabular}{p{0.8cm}p{5.5cm}p{6.6cm}p{3.2cm}}
\hline\noalign{\smallskip}
Tag & CPU & OpenCL Device & Operating System\\
\hline 
Carrizo & AMD Pro A12-8800B (2.1 GHz, 4 cores) & AMD Radeon R7, OpenCL 2.0, driver 15.30.3 & Ubuntu 14.04, g++ 4.8.4 \\
\hline 
i7 & Intel Core i7-6700HQ (2.6 GHz, 4 cores) & CPU, Intel OpenCL 1.2 driver 6.2.0.1760 & Ubuntu 16.04, g++ 4.8.5 \\
\hline 
RX & Intel Core i7-4770 (3.5 GHz, 4 cores) & AMD Radeon RX 460, OpenCL 1.2, driver 16.40 & Ubuntu 16.04, g++ 4.8.5 \\
\hline 
\noalign{\smallskip}
\end{tabular}
\end{table*}

\begin{table}
\caption{Motion Detection performance in HD 1080p frames/s.}
\label{table:motion}
\begin{tabular}{p{1.2cm}p{1.4cm}p{1.4cm}p{1.4cm}p{1.4cm}}
\hline\noalign{\smallskip}
Tag & DAL & PRUNE & DAL & PRUNE \\
 & Multicore & Multicore & Heterogen. & Heterogen. \\
\hline 
Carrizo & 15.9 & 17.3 & 132 & 140 \\
\hline 
i7 & 39.2 & 39.8 & - & 113 \\
\hline 
RX & 41.1 & 42.9 & 696 & 772 \\
\hline 
\noalign{\smallskip}
\end{tabular}
\end{table}

\begin{table}
\caption{Digital Predistortion performance in complex float megasamples/s.}
\label{table:dpd}
\begin{tabular}{p{1.2cm}p{1.4cm}p{1.4cm}p{1.4cm}p{1.4cm}}
\hline\noalign{\smallskip}
Tag & DAL & PRUNE & DAL & PRUNE \\
 & Multicore & Multicore & Heterogen. & Heterogen. \\
\hline 
Carrizo & 6.15 & 6.64 & n/a & 39.5\\
\hline 
i7 & 6.78 & 7.34 & n/a & 25.1 \\
\hline 
RX & 17.2 & 19.2 & n/a & 106\\
\hline 
\noalign{\smallskip}
\end{tabular}
\end{table}

\begin{table}
\caption{Adaptive Deep Neural Network performance in frames/s.}
\label{table:dnn}
\begin{tabular}{p{1.2cm}p{1.4cm}p{1.4cm}p{1.4cm}p{1.4cm}}
\hline\noalign{\smallskip}
Tag & DAL & PRUNE & DAL & PRUNE \\
 & Multicore & Multicore & Heterogen. & Heterogen. \\
\hline 
Carrizo & 8.63 & 11.5 & n/a & 160 \\
\hline 
i7 & 9.37 & 9.86 & n/a & 299 \\
\hline 
RX & 26.2 & 37.8 & n/a & 1033 \\
\hline 
\noalign{\smallskip}
\end{tabular}
\end{table}

\begin{figure}
\centering
\includegraphics[width=\linewidth]{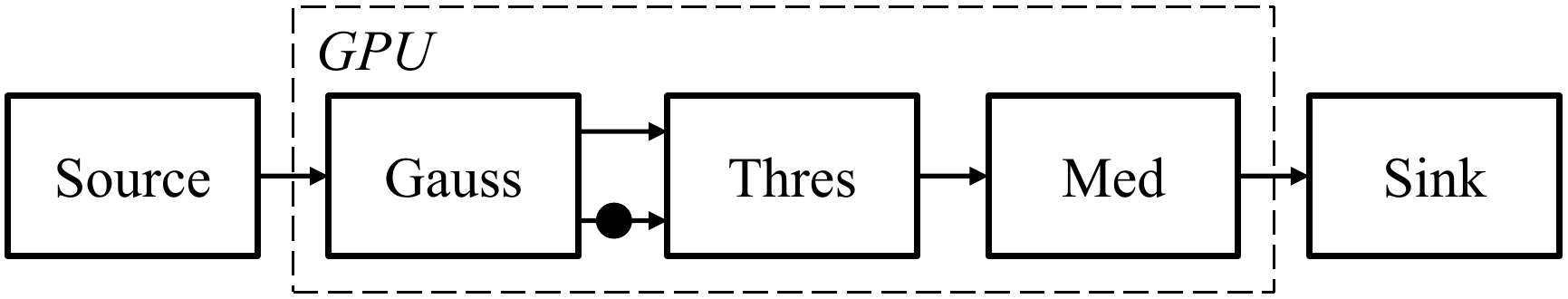}
\caption{The Motion Detection application.}
\label{fig:video}
\end{figure}

\begin{figure}
\centering
\includegraphics[width=0.90\linewidth]{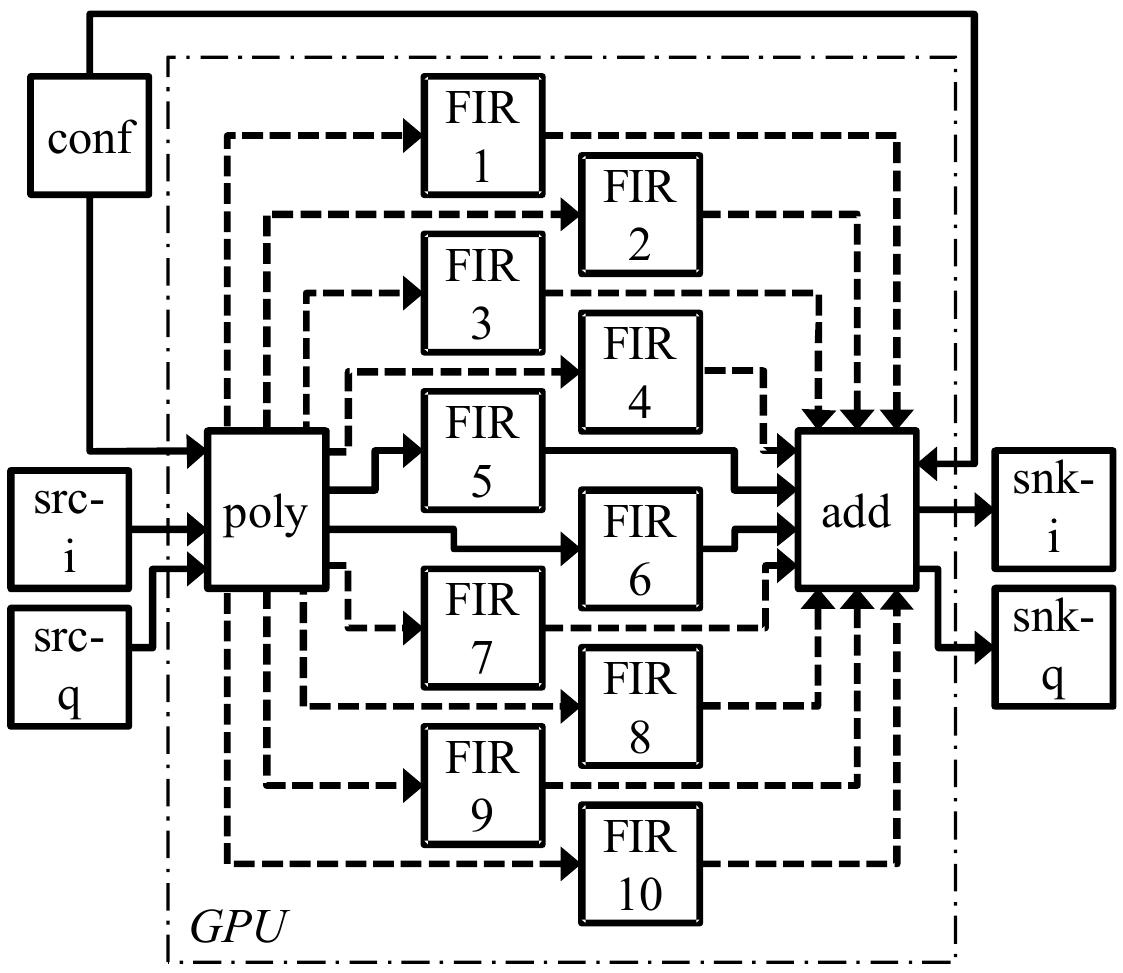}
\caption{The Dynamic Predistortion application.}
\label{fig:dpd}
\end{figure}

\section{Experimental Results}
\label{sec:experiments}

This section presents experimental evaluation, which shows that the PRUNE framework is efficient and suitable for running real-life signal processing workloads. The performance results are compared against the DAL framework \cite{Schor13}, which in many ways resembles PRUNE.

\subsection{Use Case Applications}
\label{ssec:apps}

\subsubsection{Video Motion Detection}

The first application used in our experiments is 8-bit grayscale video Motion Detection that consists of five actors, as shown in Fig.~\ref{fig:video}. The source and sink actors are always executed on CPU cores and are essentially responsible for reading and writing data from/to mass storage. The Gauss actor performs 5$\times$5 pixels Gaussian filtering on input frames, followed by the Thres actor that subtracts consecutive frames and performs pixel thresholding against a fixed constant value. To avoid exceeding frame boundaries, the Gauss actor skips filtering for two pixel rows in the frame top and frame bottom. Finally, the Med actor performs 5-pixel median filtering to reduce noise from the generated motion map.

The distinguishing feature of this use case application is the use of delay tokens: one of the communication channels between the Gauss and Thres actors bears a dot in Fig.~\ref{fig:video} and depicts an initial token, which is a one-frame delay that enables consecutive frame subtraction functionality. Out of the use case applications, Motion Detection is the only one that can be described using fixed token rates, and hence be GPU-benchmarked with both PRUNE and the DAL framework, which we use as a state-of-the-art reference.

The frame size used was 1920$\times$1080, which resulted in the token size becoming 1.98 megabytes. Due to the large token size, the token rate was kept at 1 (in our previous publication \cite{Boutellier2016} that uses a preliminary version of PRUNE, resolution was 320x240 with a token rate of 4).
GPU acceleration was
applied to Motion Detection by mapping the Gauss, Thres and Med actors to the GPU.

\subsubsection{Dynamic Predistortion Filtering}

Dynamic Predistortion (DPD) filtering (Fig.~\ref{fig:dpd}) was used as the second application use case. The algorithm \cite{Abdelaziz13} is used in wireless communications to mitigate transceiver impairments, and essentially consists of 10 parallel 10-tap complex-valued floating-point FIR filters.

DPD significantly differs from the Motion Detection application in the sense that it features actors with dynamic token rates: Fig.~\ref{fig:dpd} shows the configuration (conf) actor that at run-time periodically reconfigures the poly and Adder (add) actors to select which set of the FIR filters is used to process the input signal. The reconfiguration period was set to once every 65536 samples, and the number of active filter actors is allowed to change arbitrarily between 2 and 10. The run time reconfiguration used here is defined by an external input and cannot be modeled e.g. by the CSDF MoC.

The DPD application computes on complex floating-point numbers, which were represented as a pair of single precision floats. To this end, all edges in Fig.~\ref{fig:dpd} \textit{inside} the "GPU" box represent a pair of edges, one for the real part and one for the imaginary part. Hence, the total number of FIFO channels is 46 in this application.

\begin{figure}
\centering
\includegraphics[width=\linewidth]{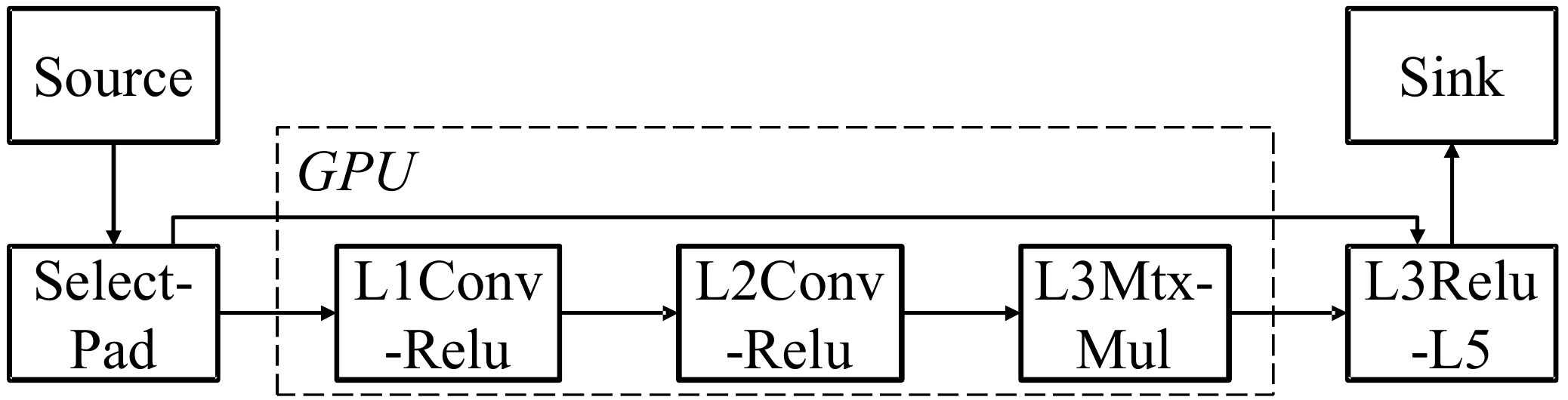}
\caption{The Adaptive Deep Neural Network application.}
\label{fig:cnn}
\end{figure}

\subsubsection{Adaptive Deep Neural Network}
\label{sssec:dnn}

The Deep Neural Network (DNN) application for vehicle classification has first been presented in \cite{hutt2016x3, xie2016resource}. The neural network consists of two convolutional layers followed by three dense layers. The PRUNE implementation of the DNN consists of seven actors, as depicted in Fig.~\ref{fig:cnn}. Here, the application has been extended with adaptiveness that allows dynamically enabling and disabling DNN processing for each frame to, e.g., save on power when deployed to an energy-limited device.

The three GPU-mapped actors (Fig.~\ref{fig:cnn}) represent the neural network core of the application and are demanding both in terms of memory footprint and computational complexity. Layer 1 convolution consists of 2400 floating-point weights, layer 2 of 25600 weights and layer 3 matrix multiplication of 1.8M weights. The input image is delivered in resolution 96$\times$96 (RGB) and separated to 32 feature maps that are convolved by 5$\times$5 pixel kernels in layers 1 and 2. 

The GPU-accelerated convolution layers perform simultaneous convolution and ReLU non-linearity computation, and are processed as a 3-dimensional volume along feature map index, image width and image height axes. Data was mapped to tokens such that the data associated with one input image was mapped to one token. Hence, due to the nature of the algorithm, the token size varied considerably from one FIFO to another. Benchmarking was performed with $\mxatr=24$ for each FIFO channel, as this token rate provided the highest throughput.

Adaptiveness was implemented to the DNN application by introducing a \textit{bypass} channel from the actor Select-Pad to the actor L3Relu-L5. With the bypass channel the computationally demanding DNN processing can be omitted for selected frames, and instead of classification results, the bypass channel provides constant values to the output to indicate omitted classification.

In order to demonstrate the possibility for performing post-analysis application optimizations, the configuration actor was merged to the dynamic actor Select-Pad, preserving identical application functionality. For the moment, the actor merging needs to be done manually, however automatic approaches exist \cite{Boutellier15T}.

Besides lack of dynamic token rates, implementation of GPU accelerated DNN in DAL turned out to be unfeasible for several reasons. First, DAL only supports 1-dimensional kernel processing (OpenCL NDRange), and second, DAL provides no direct means to deliver megabytes of fixed coefficients to GPU-accelerated kernels. For these reasons, no GPU accelerated DAL version of the application was created.

\subsection{The Platforms}
\label{ssec:platforms}

Table~\ref{table:platforms} shows the platforms that were used to benchmark the PRUNE framework and the Distributed Application Layer, which was used as a reference. The \textit{Carrizo} chip features 4 CPU cores and an integrated graphics processor, a solution that minimizes the data transfer times between the GPU and the CPU cores. \textit{RX} represents a conventional desktop system with a 4-core CPU and a mid-range GPU that is connected to the CPU over a PCI express bus. Finally, \textit{i7} represents a laptop processor with OpenCL drivers that allow accelerating data parallel workloads on the CPU cores. 

\subsection{Experimental Setup}
\label{ssec:exp-setups}

For all use case applications, the PRUNE run-time library was configured to implement triple-buffering of FIFOs ($C=3$, see Eq.~\ref{eq:fifosize}), which provided equal FIFO memory sizes as DAL. DAL applications were implemented to use high-speed \textit{windowed FIFOs} in communication between CPU cores. For each platform and each use case application the execution time was calculated from the average of 8 successive application executions. Before measurements, it was ensured that the processor cores were almost idle by closing unnecessary applications in the OS.

In Motion Detection the OpenCL global work size was set to 76800 for \textit{i7} and to 518400 for \textit{Carrizo} and \textit{RX}. The input data file was a grayscale 300 frame uncompressed sequence ``Jockey'' in resolution 1920$\times$1080, which resulted in a file size of 593 MB.

For the DPD application, OpenCL global work size was equal to the actor input token rate, which was either 65527 or 65536 depending on the actor. The same work size was used for all platforms. The DPD input data stream consisted of 67 megasamples, altogether 537 MBs of size.

For Adaptive DNN the OpenCL work size dimensions were 768$\times$52$\times$52 for L1Conv-Relu, 768$\times$24$\times$24 for L1Conv-Relu, and 24$\times$100 for L3Mtx-Mul. The input sequence consisted of 384 RGB frames that had a file size of 40.5 MB due to their single-precision float data format.

\subsection{Results}
\label{ssec:results}
Table~\ref{table:motion} shows that with the Motion Detection application the PRUNE framework provided 2-9\% higher throughput than DAL on each platform, when no OpenCL acceleration was used, but all processing was done by CPU cores (the columns labeled ``Multicore''). With OpenCL acceleration (``Heterogen.'' columns) the Motion Detection application throughput increased 3$\times$ to 18$\times$ such that PRUNE was 6\% to 11\% faster than DAL, depending on the platform. Under DAL the Intel OpenCL drivers caused an error on the \textit{i7} that prohibited benchmarking.

Multicore-only benchmarking of the Digital Predistortion application revealed 8\%-12\% higher throughput for PRUNE compared to DAL, as Table~\ref{table:dpd} shows. OpenCL acceleration increased the application performance by 3$\times$ to 6$\times$ under PRUNE compared to multicore-only. DAL OpenCL results could not be acquired, as DAL is restricted to static token rates under OpenCL.

The Adaptive Deep Neural Network application revealed (Table~\ref{table:dnn}) larger differences in throughput: under multicore-only, PRUNE was 5\% to 44\% faster than DAL. OpenCL acceleration was also remarkably powerful, as the application speeded up between 14$\times$ to 30$\times$ by OpenCL on PRUNE. As mentioned in Section~\ref{sssec:dnn}, OpenCL results could not be acquired due to several restrictions of DAL.

\section{Discussion and Future Work}
\label{sec:discussion}

The use case applications introduced in Section~\ref{ssec:apps} demonstrate that the PRUNE Model of Computation is expressive enough for describing a wide variety of performance-intensive signal processing applications, which are highly relevant to the video processing, computer vision and wireless communications fields. In contrast, the state-of-the-art framework DAL could not provide means for OpenCL acceleration of Digital Predistortion or Adaptive DNN applications, or means for analyzing the consistency of the application graphs.

The results in Section~\ref{ssec:results} show that the PRUNE runtime framework is also remarkably efficient compared to the state-of-the-art DAL framework: application performance on CPU cores is up to 44\% higher under PRUNE. PRUNE also enables highly efficient simultaneous use of OpenCL devices: performance increases up to 30$\times$ were measured compared to CPU-only.

As future work for PRUNE, interfaces for importing signal processing applications written for the DAL framework and the Open-RVC-CAL Compiler \cite{Yviquel13} will be developed. This will enable the novel capabilities for high performance signal processing in PRUNE to be leveraged by applications.

\section{Conclusion}
\label{sec:conclusion}
In this article the PRUNE Model of Computation and framework has been presented. The dataflow oriented PRUNE Model of Computation is expressive enough for describing signal processing applications with dynamic token rates, yet it provides at the same time decidable
 deadlock freedom and memory boundedness analysis of applications.

The expressiveness of PRUNE has been demonstrated by examples from three signal processing domains: computer vision, video processing and wireless communications. Experimental results have shown that besides decidability, PRUNE also provides more expressiveness and higher performance than the previously published state-of-the-art DAL framework.

\section*{Acknowledgment}

This work was partially funded by the Academy of Finland project 309693 UNICODE and by TEKES --- the Finnish Technology Agency for Innovation (FiDiPro project StreamPro 1846/31/2014).



\bibliographystyle{IEEEtran}
\bibliography{IEEEabrv,refs}
%


%

\begin{IEEEbiography}[{\includegraphics[width=1in,height=1.25in,clip,keepaspectratio]{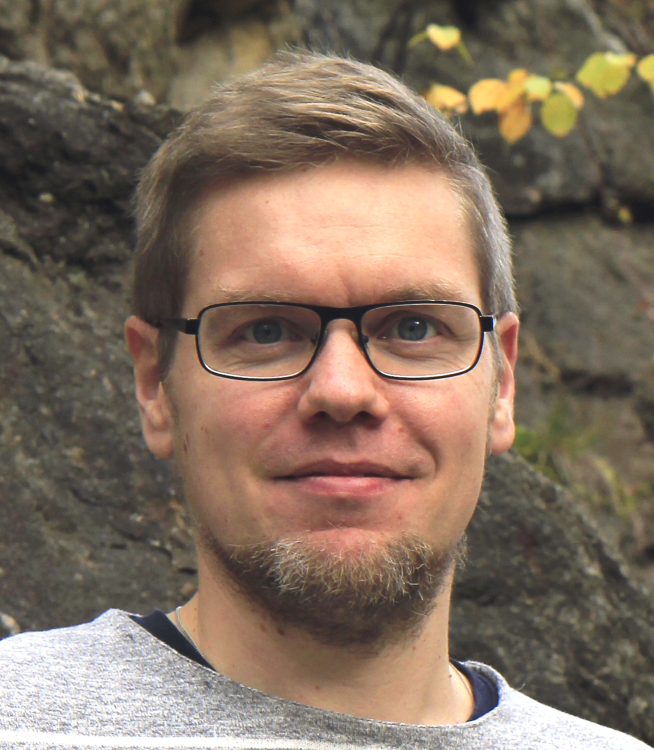}}]{Jani Boutellier} received the M.Sc. and Ph.D. degrees from the University of Oulu, Finland, in 2005 and 2009, respectively. Currently he is an Assistant Professor at the Laboratory of Pervasive Computing of Tampere University of Technology, Finland. In 2007-2008, 2013 he was working as a visiting researcher with the EPFL, Switzerland. His research interests include dataflow programming, signal processing, hardware-software codesign and heterogeneous computing. He is a member of the IEEE Signal Processing Society Design and Implementation of Signal Processing Systems Technical Committee.
\end{IEEEbiography}

\begin{IEEEbiography}[{\includegraphics[width=1in,height=1.25in,clip,keepaspectratio]{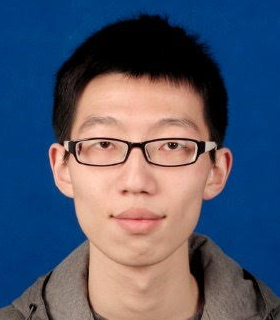}}]{Jiahao Wu} received the bachelor's degree from the University of Electronic Science and Technology of China (UESTC). He joined the Department of Electrical and Computer Engineering at the University of Maryland, College Park as a Ph.D. Student in 2014. His research interests include model-based design for parallel computing, dataflow implementations and synthesis of digital signal processing systems.
\end{IEEEbiography}

\begin{IEEEbiography}[{\includegraphics[width=1in,height=1.25in,clip,keepaspectratio]{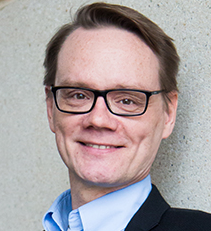}}]{Heikki Huttunen} was born in 1971 and received M.Sc. in mathematics from University of Tampere, Finland, in 1995 and Ph.D. degree in signal processing from Tampere University of Technology, Finland, in 1999. Between 2003-2005 he worked with Visy Oy, developing automatic license plate recognition systems. Since then he held various positions at Tampere University of Technology where he currently is an associate professor leading the machine learning group. His research interests are in the field of machine learning, in particular in deep learning, with research focus on efficient real-time and real-life implementations, and classifier error estimation. Dr. Huttunen has published ca. 100 journal and conference articles and he is an Associate Editor of Journal of Signal Processing Systems. 
\end{IEEEbiography}

\begin{IEEEbiography}[{\includegraphics[width=1in,height=1.25in,clip,keepaspectratio]{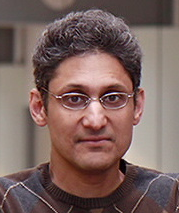}}]{Shuvra S. Bhattacharyya} is a Professor in the Department of Electrical and Computer Engineering at the University of Maryland, College Park. He holds a joint appointment in the University of Maryland Institute for Advanced Computer Studies (UMIACS). He is also a part time visiting professor in the Department of Pervasive Computing at the Tampere University of Technology, Finland, as part of the Finland Distinguished Professor Programme (FiDiPro).  He is an author of six books, and over 250 papers in the areas of signal processing, embedded systems, electronic design automation, wireless communication, and wireless sensor networks. He received the B.S. degree from the University of Wisconsin at Madison, and the Ph.D. degree from the University of California at Berkeley. He has held industrial positions as a Researcher at the Hitachi America Semiconductor Research Laboratory (San Jose, California), and Compiler Developer at Kuck \& Associates (Champaign, Illinois). He has held a visiting research position at the US Air Force Research Laboratory (Rome, New York). He has been a Nokia Distinguished Lecturer (Finland) and Fulbright Specialist (Austria and Germany). He has received the NSF Career Award (USA). He is a Fellow of the IEEE.

\end{IEEEbiography}

\end{document}